\documentclass[11pt]{llncs}

\usepackage{graphicx}
\usepackage{url}
\usepackage{amsmath,amssymb}

\makeatletter
\long\def\@caption#1[#2]#3{\par\addcontentsline{\csname
  ext@#1\endcsname}{#1}{\protect\numberline{\csname 
  the#1\endcsname}{\ignorespaces #2}}\begingroup
    \@parboxrestore
    \small
    \@makecaption{\csname fnum@#1\endcsname}{\ignorespaces #3}\par
  \endgroup}
\makeatother

\makeatletter
\leftmargin=\leftmargini
\labelwidth=\leftmargini
\advance \labelwidth by -\labelsep
\def\@listI{%
  \leftmargin=\leftmargini
  \parsep=\z@
  \topsep=2pt plus 2pt minus 2pt
  \itemsep=0pt plus 2pt\relax}%
\let \@listi=\@listI
\@listi
\makeatother

\title{Morpion Solitaire 5D: \\ a new upper bound of 121 on the maximum score}

\author{%
   Akitoshi Kawamura\inst{1} \and
   Takuma Okamoto\inst{2} \and
   Yuichi Tatsu\inst{2} \and \\
   Yushi {\rm Uno}\inst{2} \and 
   Masahide Yamato\inst{2}
}
\institute{
    Department of Computer Science, University of Tokyo, 
    7-3-1 Hongo, Bunkyo-ku, Tokyo 113-8656, Japan. 
    \email{kawamura@is.s.u-tokyo.ac.jp} 
\and
    Graduate School of Science, Osaka Prefecture University, 
    1-1 Gakuen-cho, Naka-ku, Sakai 599-8531, Japan. 
    \email{ss301002@edu.osakafu-u.ac.jp}, 
    \email{sr301023@edu.osakafu-u.ac.jp}, 
    \email{uno@mi.s.osakafu-u.ac.jp}, 
    \email{sr301036@edu.osakafu-u.ac.jp}.
}

\index{Kawamura, Akitoshi}
\index{Okamoto, Takuma}
\index{Tatsu, Yuichi}
\index{Uno, Yushi}
\index{Yamato, Masahide}

\begin{document}

\maketitle

\begin{abstract}
Morpion Solitaire is a pencil-and-paper game for a single player. 
A move in this game consists of putting a cross at a lattice point
and then drawing a line segment that passes through 
exactly five consecutive crosses. 
The objective is to make as many moves as possible, 
starting from a standard initial configuration of crosses. 
For one of the variants of this game, 
called 5D, 
we prove an upper bound of $121$ on the number of moves.
This is done by introducing \emph{line-based analysis}, 
and improves the known upper bound of $138$ 
obtained by potential-based analysis. 

\medskip
\noindent
Keywords: \emph{pencil-and-paper game, lattice points, 
line-based analysis.} 
\end{abstract}

\section{Introduction}

\emph{Morpion Solitaire}, also known as \emph{Join Five}, 
is a game 
played alone with a pencil and paper, 
and it is popular in several countries \cite{morpion_web}. 
A move in this game consists of drawing a cross and a line segment
on an infinite square lattice. 
The line segment has to pass through
exactly five consecutive crosses including the one that has just been placed. 
The objective is to make as many moves as possible 
starting from a given initial configuration. 
We call the number of moves the \emph{score}. 
There are two variants of this game 
according to how two line segments can touch each other. 

Demaine et al.~\cite{demaine} studied 
generalizations of the game and their computational complexity, 
and show that a generalized Morpion Solitaire is NP-hard 
and that its maximum score is hard to approximate. 
Another target of interest is 
the maximum scores or their lower and upper bounds. 
Recently, computing maximum scores was used as a test problem 
to evaluate the effectiveness of the Monte-Carlo tree search method, 
which has been attracting rising attention as a promising approach 
in game programming \cite{cazenave,rosin}. 

In this paper, we focus on the 5D variant of the game, 
and show improved upper bounds on the maximum score. 
We first show that the known upper bound of $138$ 
can be improved to $136$ 
by pushing on the existing potential-based approach. 
Next we introduce a line-based approach
and further improve the bound to $121$. 
We also try to organize and present related results, 
since there are relatively few research papers on this topic.

\section{Rules and Records}

\subsection{Rules}

Morpion Solitaire is played on an infinite square lattice. 
Initially 36 crosses are drawn on lattice points 
so that they form a large cross shape with edge length $4$ 
as shown in Figure~\ref{morpion_board}. 
In this figure, a cross is denoted by a circle. 
(In this paper, the length of a line segment means
the number of crosses covered by it.) 

\begin{figure}
\centering
\scalebox{0.70}{\includegraphics{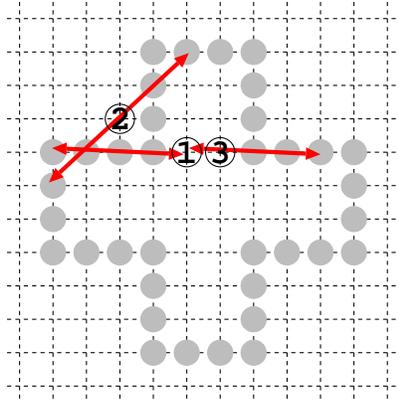}}
%\scalebox{0.65}{\includegraphics{./board.eps}}
\caption{%
The standard initial board layout for Morpion Solitaire 5D and 5T, 
and an example of the first three moves. 
Each cross placed in these moves is denoted 
by a number surrounded by a circle. 
Move 3 is allowed in 5T (touching) but not in 5D (disjoint).} 
\label{morpion_board}
\end{figure}

A {\it move} consists of the following 
two steps applied in this order. 
The objective of this game is to maximize the number of moves. 
\begin{enumerate}
\item 
Draw a new {\it cross} on a lattice point which is empty (no cross exists) 
on the current board. 
\item
Draw a segment of length 5 (called a {\it line}) 
that passes through exactly five consecutive crosses 
including the one drawn in step 1 of this move. 
Here, the line can be drawn in either one of the four directions, 
vertical, horizontal, or diagonal. 
Two lines in the same direction may not overlap. 
\end{enumerate}

There are two variants of this game depending on 
whether two lines in the same direction 
can touch (5T) or have to be disjoint (5D)
(Figure~\ref{morpion_board}). 
We mainly discuss about 5D in this paper. 

When a line $L$ passes a cross $C$, 
we say that $L$ {\it covers} the cross or the lattice point 
on which it is drawn. 
We sometimes call a board after move $N$ a board at move $N$. 
Also we sometimes denote a cross and a line drawn in move $N$ 
by $C_N$ and $L_N$, respectively.

\subsection{Records}

The above definition of the game can be extended to $\alpha$D and $\alpha$T, 
where the lines have length $\alpha$ and 
the edges of the large cross in the initial configuration 
have length $\alpha - 1$, however, 
the maximum scores are known for all variants except $\alpha = 5$.
For 3T and 3D, the maximum scores are not bounded, 
as there are sequences of moves 
that can be repeated infinitely \cite{demaine}. 
For 6T and 6D, we can easily see that 
the maximum score is $12$. 
For 4T and 4D, there used to be gaps between the maximum achieved scores 
and the upper bounds in the past, but in 2007, 
62 and 35 moves were achieved for 4T and 4D, respectively \cite{hyyro}, 
and these scores were proved to be optimal in 2008 \cite{morpion_web}. 
%at 62 for 4T and 35 for 4D. 

Table~\ref{table:record} \cite{morpion_web} shows the current maximum scores
of 5T and 5D. 
We brief\textcompwordmark ly explain how the records 
of these two variants have been developed. 

\begin{table}[b]
\centering
\caption{%
Records on Morpion Solitaire 5T and 5D: 
their maximum achieved scores 
and proven upper bounds.} 
\label{table:record}
\medskip
\renewcommand{\arraystretch}{1.0}
\begin{tabular}{|c||c|c|}
    \hline
    {\small game type} & {\small best achieved score} & {\small upper bound} \\ \hline\hline
    5T & 178 & 705 \\ \hline
    5D & 82 & 138 \\ \hline
\end{tabular}
\end{table}

\medbreak
\noindent \textbf{5T.}
Bruneau achieved 170 in 1976 by hand \cite{bruneau}. 
In 2010, by computer, Akiyama, Komiya and Kotani \cite{akiyama} 
used Monte-Carlo tree search to achieve 145 and 146, 
which were still less than human's record at that time. 
From 2010 to 2011, also by computer, 
Rosin achieved 172, 
beating human's record \cite{boyer}. 
Rosin \cite{rosin} improved the record to 177 in 2011, 
and the current record is 178 \cite{rosin_web}. 
An upper bound of $705$ on the maximum score is known~\cite{demaine}. 

\medbreak
\noindent \textbf{5D.}
According to Demaine et al.~\cite{demaine}, 
$68$ moves was achieved by hand in 1999. 
Cazenave \cite{cazenave} established $80$ in 2008, 
and then Rosin \cite{rosin} improved it to $82$ in 2010, both by computers. 
As for upper bounds, 
Demaine et al.~\cite{demaine} showed 141 in 2006 \cite{demaine} 
and Karjalainen showed 138 in 2011 \cite{karjalainen}. 

\medbreak

Recent records of maximum scores of both 5T and 5D 
were obtained by computers. 
The framework used for this was Monte-Carlo tree search or its extensions, 
which are known to produce excellent results in designing 
computer programs, for example, for playing Shogi or Go against humans. 

Hereafter, in this paper, we focus only on 5D variant 
and aim to improve the upper bound on its maximum score, 
which is known to be 138.

\section{Potential-based Analysis of Upper Bounds}

The known upper bound of 138 on the maximum score of Morpion Solitaire 5D 
is obtained by arguments using `potentials'. 
In this section, 
we explain potentials and the related results, 
and then show that the upper bound can be improved to 136 
by a more detailed analysis based on this approach. 

\subsection{Preceding Research}

The notion of potential in the analysis of Morpion Solitaire 
seems to have been originally introduced in folklore discussions 
and was used 
by Demaine et al.~\cite{demaine}. 
The \emph{potential} of a cross on a board 
is the number of additional lines that can cover it. 
Since a cross can be covered by at most four lines 
(in the vertical, horizontal and two diagonal directions), 
the potential of a cross $C$ is formally given by 
\[
4- (\mbox{number of lines that cover } C). 
\]
We define the {\it total potential} of a board 
to be the sum of the potentials of all crosses on that board. 

\renewcommand{\labelenumi}{(\theenumi)}
\renewcommand{\theenumi}{\roman{enumi}}
Now we can observe the following three facts 
about Morpion Solitaire 5D. 

\medbreak
\noindent
{\bf Observations} 
\begin{enumerate}
\item 
The total potential of the initial board is 144. 
\item
The total potential decreases at least by 1 in every move. 
\item
At any time, playing the next move requires at least a total potential 4. 
\end{enumerate}

\medbreak
\noindent
We have (i) because 
initially there are $36$ crosses, each of which has potential $4$. 
We have (ii) because 
step~1 of a move in 5D adds $4$ to the total potential, 
and step~2 decreases the potential by $5$. 

Demaine et al.~\cite{demaine} showed 
the following upper bound based on the above three observations. 

\begin{theorem}[\cite{demaine}]
The number of moves in Morpion Solitaire 5D
cannot exceed $141$. 
\end{theorem}

To see this, let $M$ be the maximum score (the number of moves). 
The total potential after $M-1$ moves 
must be at least $4$, 
that is, $144 - (M-1) \ge 4$. 

Karjalainen \cite{karjalainen} improved this argument 
and obtained the following result 
by showing that the total potential at any time 
is at least $6$. 

\begin{theorem}[\cite{karjalainen}]
\label{karjalainen138}
The number of moves in Morpion Solitaire 5D
cannot exceed $138$. 
\end{theorem}

To see this, 
let $M$ be the maximum score 
and consider the last three moves. 
The crosses drawn in the last three moves $M$, $M-1$ and $M-2$ 
are eventually covered by one line, by at most two lines, 
and by at most three lines, respectively. 
In other words, those crosses have potentials 
$3$, $\geq 2$, and $\geq 1$, respectively, 
at the end of the game. 
This implies $144-M\ge 6$, and thus $M\le 138$.

%------------------------------------------------------------
\subsection{Improvements}

We next show some small improvements of maximum scores in the framework 
of potential-based analysis. 
Our improvements are obtained by focusing on the last four moves. 
We denote the potential of a cross $C$ on a board by $p(C)$. 

\begin{lemma}
\label{last3moves7}
The sum of the potentials of the three crosses 
that are drawn in the last three moves is greater than or equal to $7$. 
\end{lemma}

\begin{proof}
Consider the board at move $N$. 
According to the arguments for Theorem~\ref{karjalainen138}, 
$p(C_N)=3$, $p(C_{N-1})\ge 2$ and $p(C_{N-2})\ge 1$ hold for crosses 
$C_N$, $C_{N-1}$ and $C_{N-2}$ at moves $N$, $N-1$ and $N-2$, respectively. 
Here, $p(C_N)=3$, $p(C_{N-1})=2$ and $p(C_{N-2})=1$ 
cannot be satisfied simultaneously. 
Suppose they can. 
Then line $L_{N-1}$ has to cover cross $C_{N-2}$ as well as $C_{N-1}$, 
and line $L_N$ has to cover both crosses $C_{N-2}$ and $C_{N-1}$ 
as well as cross $C_N$, 
and this forces such two lines $L_{N-1}$ and $L_N$ to overlap. 
This contradicts the rules of Morpion Solitaire, 
and thus $p(C_N)+p(C_{N-1})+p(C_{N-2})> 6$ holds. 
\end{proof}

Lemma~\ref{last3moves7} alone improves an upper bound to $137$, 
and we can save one more move. 

\begin{theorem}
The number of moves in Morpion Solitaire 5D
cannot exceed $136$.
\end{theorem}

\begin{proof}
Let $M$ be the maximum score, and consider a board at move $M-1$. 
First, we can see that in order that move $M$ is feasible, 
there exists a cross $C$ other than $C_{M-1}$, $C_{M-2}$ and $C_{M-3}$ 
with $p(C)\ge 1$. 
Then we determine the total potential of board $M-1$ by a case analysis; 
whether line $L_M$  drawn in move $M$ covers all three crosses 
$C_{M-1}$, $C_{M-2}$ and $C_{M-3}$, or not. 

Case 1: line $L_M$ covers all crosses $C_{M-1}$, $C_{M-2}$ and $C_{M-3}$. 
In this case, three crosses $C_{M-1}$, $C_{M-2}$ and $C_{M-3}$ lie 
on a common lattice line. 
Since no two lines can overlap, 
line $L_{M-2}$ that covers $C_{M-3}$ 
and line $L_{M-1}$ that covers both $C_{M-2}$ and $C_{M-3}$ 
are not compatible. 
Hence, $p(C_{M-1})=p(C_{M-2})=p(C_{M-3})=3$ holds. 
This, together with the fact that there exists a cross $C$ with $p(C)\ge 1$ 
other than $C_{M-1}$, $C_{M-2}$ and $C_{M-3}$ guarantees 
$p(C_{M-1})+p(C_{M-2})+p(C_{M-3})+p(C)\ge 10$. 

Case 2: line $L_M$ does not cover at least one of crosses 
$C_{M-1}$, $C_{M-2}$ or $C_{M-3}$. 
In this case, there must exist two different crosses $C$ and $C'$ 
with $p(C)\ge 1$ and $p(C')\ge 1$. 
Therefore, together with Lemma~\ref{last3moves7}, 
$p(C_{M-1})+p(C_{M-2})+p(C_{M-3})+p(C)+p(C')\ge 9$ holds. 

To put both cases together, 
the total potential of an arbitrary board of move $M-1$ 
is greater than or equal to 9. 
That is, $144-(M-1)\ge 9$ holds, which implies $M\le 136$. 
\end{proof}

\section{Line-based Analysis of Upper Bounds}

In this section, we introduce a new approach 
for deriving better upper bounds, which we call the line-based analysis. 
It is based on the relationship between the number of lines on a board 
and the number of lattice points they cover. 

The following observation is easy but crucial. 

\medskip
\noindent
{\bf Fact}\ 
After $N$ moves, 
there are $N+36$ crosses and $N$ lines. 

\medskip
\noindent
Let $c(N)$ denote the minimum number of lattice points 
that are covered by $N$ lines of length 5 in an arbitrary layout 
on a board (lattice plane). 
Then in order for a board of move $N$ to be feasible (realizable), 
it has to satisfy that $c(N)\le N+36$. 
Conversely, for $N$ that satisfies $c(N) > N+36$, 
such a move $N$ is infeasible. 
Here, since this game proceeds move by move, 
%(discretely) continuously, 
if a board of move $N$ is infeasible then all boards of moves 
greater than $N$ are infeasible. 
Hence, these observations imply the following property. 

\medskip
\noindent
{\bf Property (Board Infeasibility Condition)} 
If there exists $N$ that satisfies $c(N)> N+36$, 
then an upper bound on the maximum score is $N-1$. 

\medskip
In the subsequent discussions, we derive new upper bounds on the maximum score 
by fully utilizing this property. 
In this case, however, since it is not easy to obtain $c(N)$ directly, 
we compute a lower bound $c'(N)$ on $c(N)$, 
and we try to find $N$ that satisfies 
the Board Infeasibility Condition for that $c'(N)$.

\subsection{An Upper Bound of 132}

Here, 
we count the number of lattice points covered by lines 
by focusing on lines in one direction among four that we draw arbitrarily. 
Then we have the following lower bound on $c(N)$. 

\begin{claim}
\label{claimA}
For any move $N$, $c(N)\geq\lceil\frac{N}{4} \rceil\times 5$ holds. 
\end{claim}

\begin{proof}
Since we draw $N$ lines in all, there is a direction 
in which at least $\lceil \frac{N}{4}\rceil$ lines are drawn. 
They cover at least $\lceil \frac{N}{4} \rceil \times 5$ lattice points. 
\end{proof}

By Claim~\ref{claimA}, we have the following upper bound. 

\begin{theorem}
The number of moves in Morpion Solitaire 5D
cannot exceed $132$. 
\end{theorem}

\begin{proof}
In case that $N=133$, $c(N)\ge\lceil\frac{133}{4}\rceil\times 5=170$ holds 
according to the claim. 
On the other hand, $N+36=169$ 
and this $N$ satisfies the Board Infeasibility Condition. 
\end{proof}

\subsection{An Upper Bound of 121}

In the previous arguments, to count the number of lattice points 
covered by lines, we focused on the lines only in one direction. 
By considering two directions, 
we will obtain a tighter lower bound on $c(N)$. 
We first prove a technical lemma. 

\begin{lemma}
\label{min_cover}
Suppose that $5k+1$ $(k\ge 0)$ lines of length $5$ 
are drawn in each of two different directions (among the possible four). 
Then they cover at least $(5k+1)\times 5+4$ lattice points. 
\end{lemma}

\begin{proof}
We assume without loss of generality 
that the two different directions are vertical and horizontal. 
We color vertical lattice lines on the board periodically 
with different five colors, 
and consider the situation where $5k+1$ lines are drawn arbitrarily 
along the vertical and horizontal directions on that board. 
Notice here that the number of lattice points covered by line is the same 
in both the vertical and horizontal directions, that is $(5k+1)\times 5$. 
Then we can observe that 
\begin{enumerate}
\item
in all the lattice points covered by lines drawn in horizontal directions, 
there are exactly $5k+1$ points colored in each one of five colors, and 
\item
if we classify the lattice points covered by vertical lines by their colors, 
there are at least $5k+5$ points in some one color out of five. 
\end{enumerate}
Therefore, at least $4$ out of $5k+5$ lattice points 
are not covered by horizontal lines. 
Consequently, these lines cover $(5k +1)\times 5+4$ 
lattice points. 
\end{proof}

Figure~\ref{5k+1} shows two different layouts of lines 
where this lemma holds for $k=1$. 
%Moreover, Lemma~\ref{min_cover} can be generalized as follows 
%for different lengths of lines. 

\begin{figure}[hbt]
\centering
\scalebox{0.70}{\includegraphics{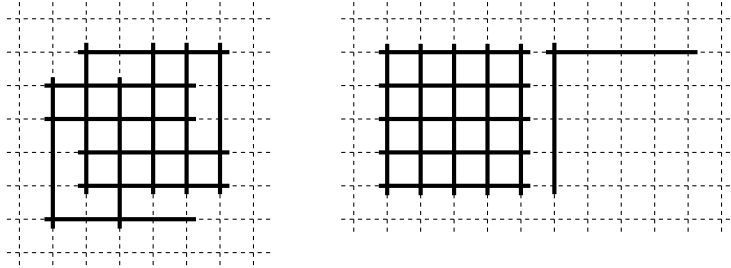}}
%\scalebox{0.64}{\includegraphics{./5k+1.eps}}
\caption{Six lines in each of two directions 
cover at least 34 lattice points.} 
\label{5k+1}
\end{figure}

%\begin{lemma}
%Suppose that $k\alpha+\beta$ $(k\ge 0; \alpha>\beta\ge 0)$ lines 
%of length $\alpha$ are drawn in each direction of two different directions 
%on board. 
%Then they cover at least $\alpha(k\alpha+\beta)+\beta\alpha-\beta^2$
%lattice points. 
%\end{lemma}

%\begin{figure*}[hbt]
%\centering
%\scalebox{0.64}{\includegraphics{./5k+1.eps}}
%\caption{Six lines in each of two directions 
%cover at least 34 lattice points.} 
%\label{5k+1}
%\end{figure*}

Now we have the following claim. 
%In the following claim, 
%we use Lemma~\ref{min_cover} with $\beta=1$. 
%That is, if we draw $5k+1$ lines in each of two different directions, 
%they cover at least $(5k+1)\times 5+4$ lattice points. 

\begin{claim}
\label{claim: b}
For a move $N$, if $N \not\equiv 1\pmod{4}$ and 
$\lceil \frac{N}{4} \rceil \equiv 1\pmod{5}$, 
then $c(N)\geq \lceil \frac{N}{4} \rceil \times 5 + 4$. 
\end{claim}

\begin{proof}
If the maximum number of lines drawn in a certain direction 
is greater than or equal to $\lceil \frac{N}{4} \rceil + 1$, 
the number of lattice points covered by some line is at least 
$\lceil\frac{N}{4}\rceil \times 5 + 5$ and the statement trivially holds. 
So suppose otherwise, that is, 
the maximum number of lines drawn in one direction 
is equal to $\lceil \frac{N}{4} \rceil$. 
Since $N \not\equiv 1\pmod{4}$, 
at least $\lceil \frac{N}{4} \rceil$ lines are drawn 
in more than one direction. 
Since this number $\lceil \frac{N}{4} \rceil$ 
equals $5k+1$ for some $k$ by assumption, 
we can apply Lemma~\ref{min_cover} 
to conclude that 
the lines drawn in these 
two directions 
cover at least $\lceil \frac{N}{4} \rceil \times 5 + 4$ lattice points. 
This implies the desired inequality. 
\end{proof}

Using this fact, we obtain a new upper bound. 

\begin{theorem}
The number of moves in Morpion Solitaire 5D
cannot exceed $121$. 
\end{theorem}

\begin{proof}
When $N=122$, 
since $122 \equiv 2\pmod{4}$ 
and $\lceil \frac{122}{4} \rceil = 1 \pmod{5}$, 
the hypothesis of Claim~\ref{claim: b} is satisfied, 
and thus $c(122)\ge 31\times 5+4=159$. 
Since this exceeds $N+36=158$, 
we have the Board Infeasibility Condition. 
\end{proof}

\begin{figure}[hbt]
\centering
\scalebox{0.70}{\includegraphics{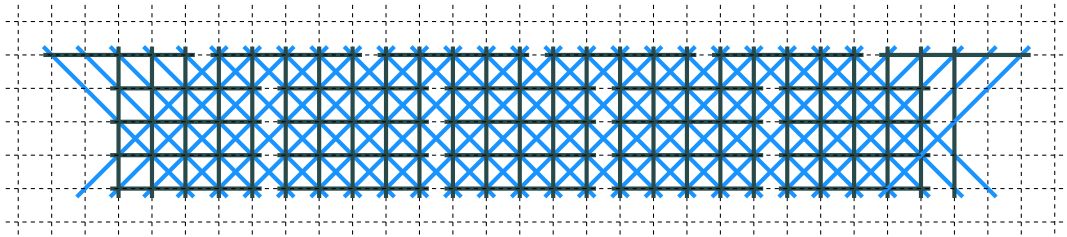}}
%\scalebox{0.65}{\includegraphics{./lower102.eps}}
\caption{102 lines cover 138 lattice points.} 
\label{lower102}
\end{figure}

\subsection{Remarks}

We mention that a similar argument to Claim~\ref{claim: b} holds
when $N \not\equiv 1 \pmod 4$ and 
$\lceil \frac{N}{4} \rceil \equiv 2$ or $3\pmod 5$. 
In this case, if the maximum number of lines drawn in a certain direction 
is $\lceil \frac{N}{4} \rceil$, 
the number of lattice points covered by some line 
is at least $\lceil \frac{N}{4} \rceil \times 5 + 6$. 
On the other hand, if the maximum number of lines drawn in a certain direction 
is equal to or greater than $\lceil \frac{N}{4} \rceil+1$, 
that is at least $\lceil \frac{N}{4} \rceil \times 5 + 5$. 
So putting these two cases together, 
we have $c(N)\geq \lceil \frac{N}{4} \rceil \times 5 + 5$. 
However, such $N$ that satisfies this hypothesis 
and the Board Infeasibility Condition is at least 126, 
and thus we know that an upper bound on the maximum score 
can be improved to 125 at best. 

We also note a limitation of this approach
of trying to use $c (N)$: 
we cannot obtain an upper bound smaller than $102$ 
by proving the Board Infeasibility Condition. 
This is because we have $c(N) \leq N+36$ for all $N\le 102$. 
Figure~\ref{lower102} proves this inequality for $N = 102$, 
and we can also easily confirm that it holds for all smaller $N$. 
We mention that a similar layout was found by Bartsch \cite{morpion_web} 
in 2010 along a different context, which tries to find a solution 
of a 5D variant.

\section{Conclusion}

Although the ultimate goal of this game is to achieve the true maximum score, 
there are some other interesting questions. 

It is possible that some idea based on our line-based analysis 
can further improve the upper bound on the maximum score of 5D. 
For example, we may consider more than two directions in those arguments. 
Also we may somehow take the initial layout of 36 crosses into account, 
which we did not in this paper. 
% Theoretically, we are interested in the real value of $c(N)$. 

We can also try to apply our line-based analysis to 5T. 
There are variants of 5T called 5T+ and 5T++, 
defined by relaxing the original rules 
about the relationship between the numbers of crosses and the lines; 
see Boyer's web page~\cite{morpion_web}. 
On this web page, 
he shows how to play $317$ moves in the 5T++ variant, 
and expects that this number may be the best possible
(and hence may give an upper bound for 5T). 
Our line-based approach may help 
prove upper bounds close to this.

\end{document}